\documentclass[12pt,a4paper]{article}

\usepackage{cancel}
\usepackage[normalem]{ulem}
\usepackage[T1]{fontenc}
\usepackage[utf8]{inputenc}
\usepackage{authblk}
\usepackage{graphicx,booktabs,balance,enumitem}
\usepackage{tikz}
\usepackage{algorithmicx}
\usepackage{algpseudocode}
\usepackage[english]{babel}
\usepackage{color}
\usepackage{graphicx}
\usepackage{hyperref}
\usepackage{amsmath,amssymb,amsthm}
\usepackage{caption}
\usepackage{listings}
\usepackage{float}

\newtheorem{defi}{Definition}

\newtheorem{teo}[defi]{Theorem}
\newtheorem{lem}[defi]{Lemma}
\newtheorem{cor}[defi]{Corollary}

\newtheorem{prop}[defi]{Proposition}
\newtheorem{ex}{Example}

\newcommand{\bbF}{{\mathbb{F}}}
\newcommand{\bbN}{{\mathbb{N}}}

\date{}

\title{Refined analysis of RGHWs of code pairs coming from Garcia-Stichtenoth's second tower\thanks{Supported by the Danish Council for Independent Research, grant DFF-4002-00367 and the Spanish Ministry of Economy, grant MTM2012-36917-C03-03.}}
\author{Olav Geil, Stefano Martin, Umberto Mart{\'\i}nez-Pe{\~n}as, Diego Ruano\thanks{Department of Mathematical Sciences, Aalborg University, 
Fredrik Bajers Vej 7G, 9220-Aalborg {\O}st, Denmark. \texttt{\{olav,stefano,umberto,diego\}@math.aau.dk}}}

\begin{document}
\maketitle

\begin{abstract}
Asymptotically good sequences of ramp secret sharing schemes were given in \cite{geil2015asymptotically} by using one-point algebraic geometric 
codes defined from asymptotically good towers of function fields. Their security is given by the relative generalized Hamming weights of the corresponding codes. In this paper we demonstrate how to obtain refined information on the RGHWs when the codimension of the codes is 
small. For general codimension, we give an improved estimate for the highest RGHW.
\end{abstract}

\section{Introduction} \label{secintroduction}

Relative generalized Hamming weights (RGHWs) of two linear codes are fundamental
 for evaluating the security of ramp secret sharing schemes and wire-tap channels of type II \cite{geil2014relative,kurihara2012secret,luo,wei}. 
Until few years ago only  the RGHWs of MDS codes and a few other examples of codes were known \cite{liu2008relative}, 
but recently new results were discovered for one-point algebraic geometric codes \cite{geil2014relative}, $q$-ary Reed-Muller codes \cite{martin2014relative}
 and cyclic codes \cite{zhang2015relative}. In \cite{geil2015asymptotically} it was discussed how to obtain
 asymptotically good sequences of ramp secret sharing schemes by using one-point algebraic geometric codes defined from 
asymptotically good towers of function fields. The tools used in \cite{geil2015asymptotically} were the Goppa bound and the Feng-Rao bounds. 
In the present paper we focus on secret sharing schemes coming from the Garcia-Stichtenoth's second tower \cite{garcia1996asymptotic}. 
We give a method for obtaining new information on the RGHWs when the used codes have small codimension. For general codimension we give an improved estimate on
the highest RGHW. The new results are obtained by studying in detail the sequence of Weierstrass semigroups related to the sequence of rational
places \cite{pellikaan1998weierstrass}.

We recall the definition of RGHWs and briefly mention their use in connection with secret sharing schemes.

\begin{defi}
Let $C_2 \subsetneq C_1 \subseteq \mathbb{F}_q^ n$ be two linear codes. For $m = 1,\ldots,\ell = \dim C_1 - \dim C_2$ the $m$-th relative generalized Hamming weight  of $C_1$ with respect to $C_2$ is 
$$M_m(C_1,C_2) = \min\{\# \mathrm{Supp}D \mid D\subseteq C_1 \mbox{ is a linear space, }$$
$$\hspace{7 cm} \dim D = m, D \cap C_2 = \{\vec{0}\}\}.$$
Here $\mathrm{Supp}D = \# \{i \in \bbN \mid \mbox{ exists }(c_1,\ldots,c_n) \in D\mbox{ with }c_i \neq 0\}$. Note that for $m = 1,\ldots, \dim C_1$, the $m$-th generalized Hamming weight (GHW) of $C_1$ $d_m(C_1)$ is $M_m(C_1,\{\vec{0}\})$. 
\end{defi}

Given $C_2 \subsetneq C_1$ linear codes, by definition, we have that the $m$-th generalized Hamming weight 
is a lower bound for the $m$-th relative generalized Hamming weight of $C_1$ with respect to $C_2$, i.e. $M_m(C_1,C_2) \geq d_m(C_1)$.
In \cite{ChenCramer}, a general construction of a linear secret sharing scheme with $n$ participants is defined from two linear codes $C_2 \subsetneq C_1$ of length $n$. It was proved in \cite{geil2014relative, kurihara2012secret} that it has $r_m = n - M_{\ell-m+1}(C_1,C_2)+1$ reconstruction and $t_m = M_m(C_2^{\perp},C_1^{\perp}) - 1$ privacy for $m =
1,\ldots,\ell$. Here, $r_m$ and $t_m$ are the unique numbers such that
the following holds: It is not possible to recover $m$ $q$-bits of
information about the secret with only $t_m$ shares, but it is possible with some
$t_m+1$ shares. With any $r_m$ shares it is possible to 
recover $m$ $q$-bits of information about the secret, but it is not possible to recover $m$ $q$-bits of information with some
 $r_m-1$ shares.

We shall focus on one-point algebraic geometric codes $C_{\mathcal{L}}(D,G)$ where $D = P_1 + \ldots + P_n$, $G = \mu Q$, and $P_1, \ldots, P_n, Q$ are pairwise different rational places over a function field. By writing $\nu_Q$ for the valuation at $Q$, the Weiestrass semigroup corresponding to $Q$ is
$$H(Q) = - \nu_Q\left( \bigcup_{\mu = 0}^\infty \mathcal{L}(\mu Q) \right) = \{\mu \in \bbN_0 \mid \mathcal{L}(\mu Q) \neq \mathcal{L}((\mu - 1) Q)\}.$$

We denote by $g$ the genus of the function field and by $c$ the conductor of the Weierstrass semigroup.
We consider $C_1 = C_{\mathcal{L}}(D,\mu_1 Q)$ and $C_2 = C_{\mathcal{L}}(D,\mu_2 Q)$, with $-1 \leq \mu_2 < \mu_1$. Observe that for $\ell = \dim(C_1) - \dim(C_2)$ and $\mu = \mu_1-\mu_2$ we have that $\ell \leq \mu$, with equality if $2g -1 \leq \mu_2 < \mu_1 \leq n-1$ holds.

  From \cite[Theorem 19]{geil2014relative} we have the following bound:

\begin{teo}  \label{teomu}
For $m = 1,\ldots,\ell$ we have that: 
$$M_m(C_1,C_2) \geq n - \mu_1 + Z(H(Q),\mu,m),$$ where
$$Z(H(Q),\mu,m) = \min\{\#\{\alpha \in \cup_{s=1}^{m-1}(i_s + H(Q)) \mid \alpha \notin H(Q)\} \mid $$
$$\hspace{6 cm}-(\mu-1) \leq i_1 < i_2 < \ldots < i_{m-1} \leq -1\}.$$
\end{teo}

For $m>g$, one has that $d_m (C) = n-k+m$, that is the Singleton bound is reached \cite[Corollary 4.2]{TsVl}. For other values of $m$, using Theorem \ref{teomu}, the following result  was found \cite[Proposition 14]{geil2015asymptotically}.

\begin{prop}  \label{propAGnew}
Let $C_{\mathcal{L}}(D,\mu Q)$ be a one-point algebraic geometric code of length $n$ and dimension $k$.
If $-1 \leq \mu < n$ and $1 \leq m \leq \min\{k,g\}$, then:
$$d_m(C_{\mathcal{L}}(D,\mu Q)) \geq n - k + 2m - c + h_{c-m} \geq n - k + 2m  - c$$ 
where $h_{c-m} = \#(H(Q) \cap (0,c-m])$.
\end{prop}

Moreover, in the proof of  \cite[Proposition 14]{geil2015asymptotically}, one has that  $$d_m(C_{\mathcal{L}}(D,\mu Q)) \geq n - \mu  + g  -1 + 2m - c + h_{c-m},$$ which may allow us to improve the bound in Proposition \ref{propAGnew} for  $\mu \le 2g -2$, since in this case  $k \ge \mu + 1 -g$. Furthermore, we can apply it to bound the RGHWs of a pair of codes.

\begin{prop} \label{propAG} 
Let $C_{\mathcal{L}}(D,\mu_2 Q) \subseteq C_{\mathcal{L}}(D,\mu_1 Q)$ be two one-point algebraic geometric codes of length $n$ and dimension $k_1$ and $k_2$, respectively.
If  $-1 \leq \mu_2 < \mu_1  < n$ and $1 \leq m \leq \min\{k_1,g\}$ then $$d_m(C_{\mathcal{L}}(D,\mu_1 Q)) \geq n - \mu_1  + g  -1 + 2m - c + h_{c-m}$$ where $h_{c-m} = \#(H(Q) \cap (0, c-m])$.
Moreover,  if $1 \leq m \leq \min\{k_1-k_2,g\}$ then 
$$M_m(C_{\mathcal{L}}(D,\mu_1 Q),C_{\mathcal{L}}(D,\mu_2 Q)) \geq n - \mu_1  + g  -1 + 2m - c + h_{c-m}$$
\end{prop}

From Garcia-Stichtenoth's second tower~\cite{garcia1996asymptotic} one obtains codes over any field $\bbF_q$ where $q$ is an even power of a
prime. Garcia and Stichtenoth analyzed the  asymptotic behaviour of the number of rational places and the genus, from which one has that 
the codes beat the Gilbert-Varshamov bound for  $q \geq 49$. This allows us to create sequences of asymptotically good codes.

The Garcia-Stichtenoth's tower $(\mathcal{F}^1,\mathcal{F}^2,\mathcal{F}^3,\ldots)$ in \cite{garcia1996asymptotic} over $\bbF_q$, for $q$ an even power of a prime, is given by: 
\begin{itemize}
	\item $\mathcal{F}^1 = \bbF_{q}(x_1) $
	\item for $\nu > 1$, $\mathcal{F}^\nu = \mathcal{F}^{\nu-1}(x_\nu)$ with $x_\nu$ satisfying $x_\nu^{\sqrt{q}} + x_\nu = \frac{x_{\nu-1}^{\sqrt{q}}}{x_{\nu-1}^{{\sqrt{q}}-1}+1}.$
\end{itemize}
The number of rational points of $\mathcal{F}^{\nu}$ is $N_{q}(\mathcal{F}^\nu) \geq q^{\frac{\nu-1}{2}}(q-{\sqrt{q}})$ and its genus is $g_\nu = g(\mathcal{F}^\nu) = (q^{\frac{1}{2}\left\lfloor\frac{\nu+1}{2}\right\rfloor}-1)(q^{\frac{1}{2}\left\lceil\frac{\nu+1}{2}\right\rceil}-1).$

For every function field $\mathcal{F}^\nu$ the following complete description of the Weierstrass semigroups corresponding 
to a sequence of rational places was given in~\cite{pellikaan1998weierstrass}. Let $Q_\nu$ be the rational point that is the unique pole in $x_1$. The Weierstrass semigroups $H(Q_\nu)$ at $Q_\nu$ in $\mathcal{F}^\nu$ are given recursively by:
\begin{eqnarray*}
H(Q_1) & = & \bbN_0 \\
H(Q_\nu) & = & {\sqrt{q}}\cdot H(Q_{\nu-1}) \cup \{i \in \bbN_0 : i \geq c_\nu\},
\end{eqnarray*}
where $c_\nu = q^{\frac{\nu}{2}} - q^{\frac{1}{2}\left\lfloor\frac{\nu+1}{2}\right\rfloor}$ is the conductor of $H(Q_\nu)$.

An alternative way to obtain these Weiestrass semigroups was described  in \cite{bras2007order}.

\begin{defi} \label{defiSM}
First we define $H(Q_1) = \bbN_0$. For $\nu$ positive integer and $j = 2\left\lfloor \frac{\nu}{2} \right \rfloor$, we define:
$$c_\nu = q^{\frac{\nu}{2}} - q^{\frac{1}{2}(\nu-\frac{j}{2})}, H(Q_\nu) = S_\nu^0 \cup S_\nu^1 \cup S_\nu^2 \cup \ldots \cup S_\nu^{\frac{j}{2}} \cup S_\nu^{\infty},$$
where:
\begin{itemize}
\item $S_\nu^0 = \{x_{\nu,1}\} = \{0\} $
\item For $1 \leq i \leq \frac{j}{2}$, $S_\nu^i = \{x_{\nu,q^{\frac{i-1}{2}}+1}, x_{\nu,q^{\frac{i-1}{2}}+2}, x_{\nu,q^{\frac{i-1}{2}}+3}, \ldots, x_{\nu,q^{\frac{i}{2}}}\}$ where for $1 \leq k \leq q^{\frac{i}{2}} - q^{\frac{i-1}{2}}$ we have that $x_{\nu,q^{\frac{i-1}{2}}+k} = q^{\frac{j}{2}}-q^{\frac{\nu-i+1}{2}}+kq^{\frac{\nu-2i+1}{2}}$ 
\item $S_\nu^{\infty} = [c_\nu+1,\infty)$
\end{itemize}
\end{defi}

Using the previous description of the Weierstrass semigroup $H(Q_\nu)$, we can see that it has the following properties:

\begin{lem} \label{rem1}
With the same notation as before, one has that:
\begin{enumerate}
\item For any $i_1,i_2 \in \{0,1,2,\ldots,\frac{j}{2}-1,\frac{j}{2}, \infty\}$, $i_1 \neq i_2$, we have that $S^{i_1}_\nu \cap S^{i_2}_\nu = \emptyset$.
\item For $i \in \{1,\ldots,\frac{j}{2}\}$ we have that $\# S^i_\nu = q^{\frac{i}{2}}-q^{\frac{i-1}{2}}$ and
 $\#\left(\cup_{r=0}^i S^r_\nu\right) = q^{\frac{i}{2}}$.
\item For $i \in \{1,\ldots,\frac{j}{2}\}$ and for any two consecutive elements $x,y \in S^i_\nu$, with $x > y$, we have that $x - y = q^{\frac{\nu-2i+1}{2}}$.
\item For $i \in \{1,\ldots,\frac{j}{2}\}$. Let $x$ be the first element of $S^i_\nu$ and $y$ the last element of $S^{i-1}_\nu$, we have that $x - y = q^{\frac{\nu-2i+1}{2}}$.
\item For $i \in \{1,\ldots,\frac{j}{2}\}$, and for any $x,y \in \cup_{r=0}^i S^r_\nu$, $x > y$ we have that $x-y \geq q^{\frac{\nu-2i+1}{2}}$. 
\end{enumerate}
\end{lem}

\begin{proof}
\begin{enumerate}
\item By Theorem 1 in \cite{bras2007order}.
\item Let $i \in \{1,\ldots,\frac{j}{2}\}$, the cardinality of $S^i_\nu$ follows by its definition.
For the second part, by (1), we have that:
$$
\#\left(\cup_{r=0}^i S^r_\nu\right) = \# S^0_\nu + \sum_{r=1}^i \# S^r_\nu = 1 + \sum_{r=1}^i (q^{\frac{r}{2}}-q^{\frac{r-1}{2}}) = 1 + q^{\frac{i}{2}} - 1 = q^{\frac{i}{2}}.
$$
\item Consider two consecutive elements $x, y \in S^i_\nu$, $x > y$. There exists a $k \in \{1,\ldots,q^{\frac{i}{2}} - q^{\frac{i-1}{2}}-1\}$ such that $x = q^{\frac{j}{2}} - q^{\frac{\nu-i+1}{2}} +kq^{\frac{\nu-2i+1}{2}}$ and $y = q^{\frac{j}{2}} - q^{\frac{\nu-i+1}{2}} +(k+1)q^{\frac{\nu-2i+1}{2}}$. It follows that $x-y=q^{\frac{\nu-2i+1}{2}}$. 
\item Let $y$ be the last element of $S^{i-1}_\nu$, i.e. $y = q^{\frac{j}{2}} - q^{\frac{\nu-i+1}{2}}$, and $x$ be the first element of $S^{i}_\nu$, i.e. $x = q^{\frac{j}{2}} - q^{\frac{\nu-i+1}{2}} + q^{\frac{\nu-2i+1}{2}}$. We have that $x - y =q^{\frac{\nu-2i+1}{2}}$.
\item For $i \in \{1,\ldots,\frac{j}{2}\}$, consider $x,y \in \cup_{r=0}^i S^r_\nu$ with $x > y$. This mean that there exists $i_1,i_2 \in \{1,\ldots,i\}$, $i_1 \geq i_2$ such that $x \in S_\nu^{i_1}$, $y \in S_\nu^{i_2}$. Let $x_2$ be the element that precedes $x$ in $H(Q^\nu)$, then we have that $x-y = (x-x_2) + (x_2-y) = q^{\frac{\nu-2i_1+1}{2}} + (x_2-y) \geq q^{\frac{\nu-2i_1+1}{2}} \geq q^{\frac{\nu-2i+1}{2}}$. The second inequality follows by (3) and (4), the third one since $x_2 \geq y$ and the last one since $i_1 \leq i$.
\end{enumerate}
\end{proof}

Applying Proposition~\ref{propAGnew} to code pairs coming from Garcia-Stichtenoth's second tower~\cite{garcia1996asymptotic}, an asymptotic result was given in~\cite[Theorem 23]{geil2015asymptotically}, which combined with Proposition \ref{propAG} allows us to obtain the following result.

\begin{cor} \label{coroAG}
Let $(\mathcal{F}_i)^{\infty}_{i=1}$ be Garcia-Stichtenoth's second tower of function fields over $\bbF_q$, where $q$ is an even power of a prime. Let $(C_i)^{\infty}_{i=1}$ be a sequence of one-point algebraic geometric codes constructed from $(\mathcal{F}_i)^{\infty}_{i=1}$. Consider $\tilde{R}, R, \rho$ with $0 \leq R \leq 1 - \frac{1}{\sqrt{q}-1}$, $0 \leq \tilde{R} < 1 $ and $0 \leq \rho \leq \min\{R, \frac{1}{\sqrt{q}-1}\}$, and assume that $\dim(C_i)/n_i \rightarrow R$ and $\mu_i/n_i \rightarrow \tilde{R}$. For all sequences of positive integers $(m_i)_{i=1}^{\infty}$ with $m_i/n_i \rightarrow \rho$, it holds that $\delta = \liminf_{i\rightarrow \infty} d_m(C_i)/n_i$ satisfies 
\begin{equation}\label{eq1}\delta \geq 1 - R + 2\rho - \frac{1}{\sqrt{q}-1}.\end{equation}
\begin{equation}\label{eq2}\delta \geq 1 - \tilde{R} + 2\rho. \end{equation}
\end{cor}

Note that the bound  (\ref{eq2}) is sharper than (\ref{eq1}) for $\frac{1}{\sqrt{q}-1} \leq R \leq 1 - \frac{1}{\sqrt{q}-1}$.

\section{Small codimension}

In this section we give a refined bound on the RGHWs of two nested one-point algebraic geometric codes coming from Garcia-Stichtenoth's towers when the codimension is small.

Before giving such bound, we illustrate the main idea with an example.

\begin{ex} \label{ex1}
Consider $q = 9$ and let $\mathcal{F}_6$ be the $6$-th function field defined by the  Garcia-Stichtenoth's tower over $\bbF_q$. The Weierstrass semigroup $H(Q_6)$ at $Q_6$ in $\mathcal{F}_6$ is
$$H(Q_6) = \{0, 243, 486, 513, 540, 567, 594, 621, 648\} \cup$$ $$ \cup \{3n \mid n \in \bbN\mbox{ and }3n \in [654, 702]\} \cup \{n \in \bbN \mid n > 702\}.$$
We denote these three sets as $A^0$, $B^0$ and $C^0$ respectively. 

For computing $Z(H(Q_6),\mu,m)$ one should find $i_1,\ldots,i_{m-1}$ such that $-(\mu-1) \leq i_1 < i_2 < \cdots < i_{m-1} \leq -1$ and minimize $\#\{\alpha \in \cup_{s=1}^{m-1} (i_s + H(Q_6)) \mid \alpha \notin H(Q_6)\}$.
In this example we fix $i_1 = -20$, thus:
$$i_1 + H(Q_6) = \{-20,223,466,493,520,547,574,601,628\} \cup $$
$$ \cup \{3n-20 \mid n \in \bbN\mbox{ and }3n \in [654, 702]\} \cup \{n \in \bbN \mid n > 682\} =$$
$$ = (i_1 + A^0)f \cup (i_1 + B^0) \cup (i_1 + C^0).$$
Note that $i_1 + A^0$ and $H(Q_6)$ are disjoint since $-i_1 = 20 < 27$ and $|x-y| \geq 27$ for any $x,y \in A^0$ $x \neq y$. For the same reason for any $-20 < i_2 < \cdots < i_{m-1} \leq -1$, we have that $i_{m-1} + A^0, i_{m-2} + A^0, \ldots , i_2 + A^0, i_1 + A^0$ and $H(Q_6)$ are disjoint. It follows that $\cup_{v=1}^{m-1} (i_v + A^0) \subseteq \{\alpha \in \cup_{v=1}^{m-1} (i_v + H(Q_6)) \mid \alpha \notin H(Q_6)\}$ and $\# \cup_{v=1}^{m-1} (i_v + A^0) = \sum_{v=1}^{m-1} \#(i_v + A^0) = (m-1)\#A^0 = 9(m-1)$. 

The same argument does not hold for $i + B^0$ (or $i + C^0$) because there exists $x,y \in B^0$ (or $C^0$) such that $|x-y| = 3$ and $-i_1 > 3$ thus it is possible that $(i_v + B^{0}) \cap B^0 \neq \emptyset$ (or $(i_v + C^{0}) \cap C^0 \neq \emptyset$) for some $v=1,\ldots,{m-1}$.

Note that $\#((i_1 + C^0) \backslash C^0) = i_1 = 20$, but  $(i_1 + C^0) \backslash C^0$ may intersect  $A^0 \cup B^0$. Therefore we consider $(i_1 + C^0) \backslash \sqrt{q}\bbN$, in this way $(i_1 + C^0) \backslash (C^0 \cup \sqrt{q}\bbN)$ and $H(Q_6)$ are disjoint.
It follows that if $i_1 = -20$, then $Z(H(Q_6),\mu,m) \geq \# \cup_{v=1}^{m-1} (i_v + A^0) + \# ((i_1 + C^0) \backslash (C^0 \cup \sqrt{q}\bbN)) = (m-1)\cdot 9 + \left\lfloor20 \frac{2}{3}\right\rfloor$. 
\end{ex}

As we can see from previous example, we do not consider the sets $B^0$, $i_1 + B^0$, $\ldots$, $i_{m-1}+B^0$ because of their intersections with other sets. 
In general, we will also consider all the possible values $-i_1$ in the range $[m-1,\mu-1]$ to obtain the following bound.

\begin{teo} \label{propemme}
Let $\nu$ be an even positive integer and $q$ an even power of a prime. Consider two one-point algebraic geometric codes $C_2 = C_{\mathcal{L}}(D,\mu_2Q) \subsetneq C_1  = C_{\mathcal{L}}(D,\mu_1Q)$ of length $n$ built on the $\nu$-th Garcia-Stichtenoth's function field over $\bbF_q$ and $\mu = \mu_1-\mu_2$. For $\mu < q^{\frac{\nu+1}{2}}$, $m=1,\ldots,\mu$, consider $u^* =  \frac{2}{3}\left(1+\frac{\nu}{4}+\log_q\left(\frac{m-1}{2(\sqrt{q}-1)}\right)\right)$ and $\beta = \min\{ 2q^{-\frac{\nu+1}{4}}(\sqrt{q}-1)(\mu-1)^{\frac{3}{2}}+1, \frac{1}{4} q^{\frac{\nu-5}{2}}(\sqrt{q}-1)^{-2}+1\}$, we have that:

$$M_{m}(C_1,C_2) \geq n - \mu_1 + g(m)$$
where
$$g(m) = \begin{cases} 
\min\big\{(m-1)q^{\frac{\nu}{4}-\frac{u}{2}} + q^{u-\frac{1}{2}}(1-q^{-\frac{1}{2}}) - 1 : 	& \mbox{ if }m > \beta \\
\hspace{0.5cm} u \in \{ \log_q(m-1)-\frac{1}{2}, \log_{q}(\mu-1)+\frac{3}{2}\}\big\} 			& \\
(m-1)q^{\frac{\nu}{4}-\frac{u^*}{2}} + q^{u^*-\frac{1}{2}}(1-q^{-\frac{1}{2}}) - 1 			& \mbox{ if }m \leq \beta, m \neq 1 \\
0																								& \mbox{ if }m = 1  
\end{cases}
$$
\end{teo}

\begin{proof}
 
By Theorem \ref{teomu} we have that $M_m(C_1,C_2) \geq n - \mu_1 + Z(H(Q_{\nu}),\mu,m)$ thus we will estimate $Z(H(Q_{\nu}),\mu,m)$. If $m=1$, $Z(H(Q_{\nu}),\mu,1)=0$, otherwise we denote the conductor of $H(Q_{\nu})$ by $c$.
Set $-(\mu-1)\leq i_1 < \cdots < i_{m-1} \leq -1$, we define $u(i_1) = \left \lfloor \log_{q}(-i_{1}) + \frac{1}{2} \right \rfloor$, then $q^{u(i_1)-\frac{1}{2}} \leq -i_1 < q^{u(i_1)+\frac{1}{2}}$. For the sake of simplicity we write $u$ instead of $u(i_1)$.  

To estimate $Z(H(Q_{\nu}),\mu,m)$ we consider the following two sets: $A(i_1,i_2,\ldots,i_v) = \{\alpha \in \cup_{v=1}^{m-1}(i_v+A^0(u))\mid \alpha \notin H(Q_\nu)\}$ where $A^0(u) = \bigcup_{i=0}^{\frac{\nu}{2}-u} S^i_\nu$ and $C(i_1) = (i_1 + C^0) \backslash H(Q_{\nu})$ where $C^0 = \{\alpha \in \bbN \mid \alpha > c\}$. Again, to simplify the notation we  write $A = A(i_1,i_2,\ldots,i_v)$, $A^0 = A^0(u)$ and $C = C(i_1)$.
By construction $A \cup C \subseteq \{\alpha \in \cup_{s=1}^{m-1}(i_s + H(Q_{\nu})) \mid \alpha \notin H(Q_{\nu})\}$ and $A \cap C = \emptyset$. Thus we have that $Z(H(Q_{\nu}),\mu,m) \geq \# A + \# C$. 

We start by computing the cardinality of $A$. By definition of $A^0$ for any $x, y \in A^0$, $x \neq y$ there exist $i_x, i_y \in \{0,\ldots,\frac{\nu}{2}-u\}$ such that $x \in S^{i_x}_\nu$ and $y \in S^{i_y}_\nu$. We can assume without loss of generality that $i_x \geq i_j$ and $x > y$, then we obtain by (6) in Lemma \ref{rem1} that $x - y \geq q^{\frac{\nu-2i_x+1}{2}}$. Since $\mu < q^{\frac{\nu+1}{2}}$, then $i_x \leq \frac{\nu}{2}-u \leq 0$ and $|x - y| \geq q^{\frac{\nu-2i_x+1}{2}} \geq q^{\frac{\nu-2(\frac{\nu}{2}-u)+1}{2}} = q^{u+\frac{1}{2}}$. Thus for $x,y \in A^0$, $x > y$, we have that  $x - y \geq q^{u+\frac{1}{2}}$.

Since $-i_1 < q^{u+\frac{1}{2}}$, it follows that $(j_1 + A^0) \cap (j_2 + A^0) = \emptyset$ for any $j_1, j_2 \in [i_{1},0]$.
Therefore we have that $\#A = \#\bigcup_{v=1}^{m-1}(i_v+A^0) = (m-1)\#A^0$. By (2) in Lemma \ref{rem1}, we have $\#A^0 = \#(\bigcup_{i=0}^{\frac{\nu}{2}-u} S^i_\nu) = q^{\frac{\nu}{4} - \frac{u}{2}}$. Thus  $\#A = (m-1)q^{\frac{\nu}{4} - \frac{u}{2}}$.

Furthermore, $\#C = \#((i_1 + C^0) \backslash H(Q_{\nu})) = \#([c+i_{1},c) \backslash H(Q_\nu)) \geq \#([c+i_{1},c) \backslash \sqrt{q}\bbN) = \left \lfloor -i_{1}(1-q^{-\frac{1}{2}}) \right \rfloor $ where the inequality follows since $H(Q_\nu) \cap [0,c) \subset \sqrt{q}\bbN$.
Hence,
\begin{eqnarray*}
M_{m}(C_1,C_2) & \geq & n - \mu_1 + Z(H(Q_{\nu}),\mu,m) \\
& \geq & n - \mu_1 + \min_{i_{1} \in \{-(\mu-1),\ldots,-(m-1)\}} (\#A + \#C) \\
& \geq & n - \mu_1 + \min \big\{(m-1)q^{\frac{\nu}{4} - \frac{u}{2}} +  \left \lfloor-i_1\left(1-q^{-\frac{1}{2}}\right)\right \rfloor \mid\\
& &  \mid i_{1} \in \{-(\mu-1),\ldots,-(m-1)\}\big\}.  
\end{eqnarray*}

One could try to minimize the previous expression bounding $u$ by $\log_{q}(-i_{1}) + \frac{1}{2}$. However, the obtained bound is too loose. Hence, we consider the minimum among all possible values of $u$ instead of $i_1$:
\begin{eqnarray*}
M_{m}(C_1,C_2) & \geq & n - \mu_1 + \min \bigg\{(m-1)q^{\frac{\nu}{4} - \frac{u}{2}} +  \left \lfloor q^{u-\frac{1}{2}}\left(1-q^{-\frac{1}{2}}\right)\right \rfloor \mid\\
 & &   {\hspace{-2 cm} \mid u \in \left\{ \left \lfloor \log_q(m-1)+\frac{1}{2} \right \rfloor, \left \lfloor \log_q(m)+\frac{1}{2} \right \rfloor , \ldots, \left \lfloor \log_q(\mu-1)+\frac{1}{2} \right \rfloor \right\} \bigg\}} \\
& \geq & n - \mu_1 + \min \bigg\{(m-1)q^{\frac{\nu}{4} - \frac{u}{2}} +  q^{u-\frac{1}{2}}\left(1-q^{-\frac{1}{2}}\right) - 1 \mid\\
& &  {\hspace{-2 cm} \mid u \in \left\{ \left \lfloor \log_q(m-1)+\frac{1}{2} \right \rfloor, \left \lfloor \log_q(m)+\frac{1}{2} \right \rfloor , \ldots, \left \lfloor \log_q(\mu-1)+\frac{1}{2} \right \rfloor \right\} \bigg\}} \\
& \geq & n - \mu_1 + \min \bigg\{(m-1)q^{\frac{\nu}{4} - \frac{u}{2}} +  q^{u-\frac{1}{2}}\left(1-q^{-\frac{1}{2}}\right) - 1 \mid\\
& & \mid \log_q(m-1)-\frac{1}{2} \leq u \leq \log_q(\mu-1)+\frac{1}{2} \bigg\}, \\
\end{eqnarray*}where the second to last inequality is obtained since $-i_1 \geq q^{u-\frac{1}{2}}$. 
We define $f(u) = (m-1)q^{\frac{\nu}{4} - \frac{u}{2}} + q^{u-\frac{1}{2}}\left(1-q^{-\frac{1}{2}}\right) - 1$. In this way our bound becomes: 
$$M_{m}(C_1,C_2) \geq n - \mu_1 + \min\left\{f(u) \mid  \log_q(m-1)-\frac{1}{2} \leq u \leq \log_q(\mu-1)+\frac{1}{2}\right\}.$$
By looking the derivative of $f(u)$, one can see that $f(u)$ only has a minimum at $u^* = \frac{2}{3}\left(1+\frac{\nu}{4}+\log_q\left(\frac{m-1}{2(\sqrt{q}-1)}\right)\right)$. However, 
it does not always hold that $\log_q(m-1)-1/2 \leq u^* \leq \log_q(\mu-1)+1/2$. This happens when either $u^* < \log_q(m-1)-\frac{1}{2}$ or $u^* > \log_q(\mu-1)+\frac{1}{2}$. The first case is equivalent to $m > \frac{1}{4} q^{\frac{\nu-5}{2}}(\sqrt{q}-1)^{-2}+1$, the second one to $m > 2q^{-\frac{\nu+1}{4}}(\sqrt{q}-1)(\mu-1)^{\frac{1}{2}}+1$.
Thus if $m > \beta = \min\{ 2q^{-\frac{\nu+1}{4}}(\sqrt{q}-1)(\mu-1)^{\frac{3}{2}}+1, \frac{1}{4} q^{\frac{\nu-5}{2}}(\sqrt{q}-1)^{-2}+1\}$, then the minimum is reached
in $\log_q(m-1)-\frac{1}{2}$ or $\log_q(\mu-1)+\frac{1}{2}$.
\end{proof}

The previous result has an asymptotic implication as well.

\begin{cor} \label{cormu}
Let $q$ be an even power of a prime, $0 \leq \tilde{R}_2 \leq \tilde{R}_1 < 1$, and $\tilde{R} = \tilde{R}_1 - \tilde{R}_2 < \frac{1}{\sqrt{q}-1}$. There exists a sequence of pairs of one-point AG codes $C_{2,i} = C_{\mathcal{L}}(D_i,\mu_{2,i}Q) \subsetneq C_{1,i} = C_{\mathcal{L}}(D_i,\mu_{1,i}Q)$, such that:
$n_i = n(C_{2,i}) = n(C_{1,i}) \rightarrow \infty$, $\mu_{j,i} / n_i $ $
\rightarrow \tilde{R}_j$  when $i \rightarrow \infty$, for $j = 1,2$. For
a given $\rho$ let 
$m_i$  be such that $m_i / n_i \rightarrow \rho$ when $i \rightarrow
\infty$  and let  $M = \liminf M_{m_i}(C_{1,i},C_{2,i}) / n_i$. It holds that:

\begin{equation*}M \geq 1 - \tilde{R}_1 + g(\rho),\end{equation*}
where
\begin{equation*}
g(\rho) = \begin{cases}
\min_{w \in \{\rho, \tilde{R}\}}\left\{\rho (w(q-\sqrt{q}))^{-\frac{1}{2}} + \frac{w}{q}(q-\sqrt{q})\right\}  & \mbox{ if }\rho > \beta, \\
\left(\frac{2\rho^2}{q}\right)^{\frac{1}{3}} + \frac{1}{\sqrt{q}} \left(\frac{\rho}{2}\right)^{\frac{2}{3}}  & \mbox{ if }\rho \leq \beta, \rho \neq 0, \\
0 		& \mbox{ if } \rho = 0,
\end{cases}
\end{equation*}
and $\beta = \min\left\{\frac{1}{4}q^{-\frac{5}{2}}(\sqrt{q}-1)^{-3}, 2q^{-\frac{1}{4}} (R\sqrt{q}-R)^{\frac{3}{2}}\right\}$.
\end{cor}

\begin{proof}
Consider the Garcia-Stichtenoth's tower $(\mathcal{F}_1,\mathcal{F}_2,\ldots)$ over $\bbF_q$ described at the end of section \ref{secintroduction}, and $0 \leq \mu_{2,i} < \mu_{1,i} \leq n_i - 1$ with  $\mu_{j,i} / n_i \rightarrow \tilde{R}_j$ for $j = 1,2$. Now consider $C_{j,i} = C_{\mathcal{L}}(D_i,\mu_{j,i}Q)$ for $j = 1,2$, where $D_i$ is a divisor of degree $n_i - 1$ and with $n_i - 1$ distinct places not containing $Q_i$, which is the unique pole of $x_1 \in \mathcal{F}_i$.
By taking the limit of the bound obtained in Theorem \ref{propemme}, the corollary holds.\\
\end{proof}

Note that if we assume that $C_{2,i}$ is the zero code for all $i$, then $\liminf M_{m_i}(C_{1,i},\{\vec{0}\})$ is the asymptotic  value of the $m_i$-th general Hamming weight of $C_{i,1}$. For $\tilde{R} < \frac{1}{4(q-\sqrt{q})}$, the bound in Corollary \ref{cormu} is sharper than the one obtained in \cite[Theorem  23]{geil2015asymptotically}.

In the following graph we compare the bound from Corollary~\ref{coroAG} (the dashed curve) with the bound from
Corollary~\ref{cormu} (the solid curve). The first axis represents $\rho = \lim m_i / n_i$, and the second axis represents 
$\delta = \liminf M_{m_i}(C_{1,i},\{\vec{0}\})$.

\begin{figure}[H]
\begin{center}
\includegraphics[scale=0.4]{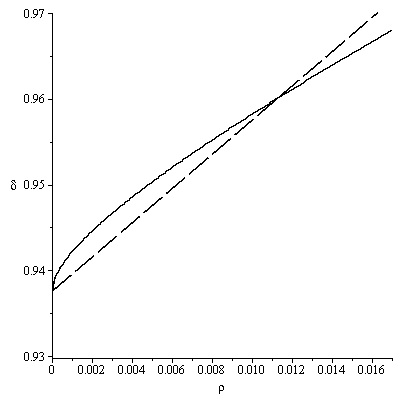}
\end{center}
\end{figure} 

\section{The highest RGHW} \label{highRGHW}

As we illustrated at the beginning of section \ref{secintroduction}, for any $n - M_{\ell}(C_1,C_2) + 1$ obtained shares an eavesdropper may recover at least one $q$-bit of the secret. In this section, for $2g -1 \leq \mu_2 < \mu_1 \le n-1$, we obtain a refined bound for the highest RGHW of two one-point algebraic geometric codes obtained from Garcia-Stichtenoth's towers, i.e. $M_{\ell}(C_1,C_2)$.

\begin{prop} \label{propmu}
Let $\nu$ be an even positive integer and $2g -1 \leq \mu_2 < \mu_1 \le n-1$. Consider two one-point algebraic geometric codes $C_2 = C_{\mathcal{L}}(D,\mu_2Q) \subsetneq C_1  = C_{\mathcal{L}}(D,\mu_1Q)$ built on the $\nu$-th Garcia-Stichtenoth's tower. We have that $\mu = \mu_1 -\mu_2 =\dim (C_1) - \dim (C_2) = \ell$ and
 
\begin{equation*}M_{\mu}(C_1,C_2) \geq n - \dim C_2 \quad \mbox{ if } \mu \geq q^{\frac{\nu-1}{2}},\end{equation*}
\begin{equation*}M_{\mu}(C_1,C_2) \geq n - \dim C_2 - \bigg(q^{\frac{\nu-1}{2}}\sum_{i=1}^{\lfloor \frac{\nu+1}{2} - \log_q(\mu) \rfloor-1}(q^{1-\frac{i}{2}} - q^{-\frac{i}{2}}) + \end{equation*}
\begin{equation*}\  \  \  \ + (q^{\frac{\nu+1}{2}-\lfloor \frac{\nu+1}{2} -  \log_q(\mu) \rfloor}-\mu)q^{\frac{\lfloor \frac{\nu+1}{2} -  \log_q(\mu) \rfloor}{2}}\bigg) \quad \mbox{ if }\mu < q^{\frac{\nu-1}{2}}.
\end{equation*}
 
\end{prop}

\begin{proof}
Since $2g -1 \leq \mu_2 < \mu_1 \le n-1$, then $\mu = \mu_1 - \mu_2 = \ell$ \cite[Lemma 12]{geil2015asymptotically}. By Theorem \ref{teomu}, we have that $M_{\mu}(C_1,C_2) \geq n - \mu_1 + Z(H(Q),\mu,\mu)$, where 
$$Z(H(Q),\mu,\mu) = \#\{\alpha \in \cup_{v=1}^{\mu-1}(-v+H(Q))\mid \alpha \notin H(Q)\}.$$
For any $x,y \in H(Q)$, we have that $|x - y| \leq q^{\frac{\nu-1}{2}}$, thus if $\mu \ge  q^{\frac{\nu-1}{2}}$ then $(\cup_{v=0}^{\mu-1} - v + H(Q)) = \bbN_0 \cup \{-1,\ldots,-(\mu-1)\}$. It follows that $Z(H(Q),\mu,\mu) = (\bbN_0 \backslash H(Q^\nu)) \cup \{-1,\ldots,-(\mu-1)\}) = g + \mu - 1$. Thus $M_{\mu}(C_1,C_2) \geq n - \mu_1 +  g + \mu - 1$. Moreover since $\mu_2 \ge  2g-1$, then $\mu_2-g+1=\dim C_2$ and the first part of the proposition holds. 

For $\ell \leq q^{\frac{\nu-1}{2}}$ we claim that:
$$Z(H(Q),\mu,\mu)=\mu+g-1 - (q^{\frac{\nu-1}{2}} \sum_{i=1}^{u_1(\mu)-1}(q^{1-\frac{i}{2}}+q^{-\frac{i}{2}})+u_2(\mu)q^{\frac{u_1(\mu)}{2}}),$$
 where $u_1(\mu) = \left\lfloor \frac{\nu+1}{2} - \log_q(\mu) \right\rfloor$ and $u_2(\mu) = q^{\frac{\nu-2u_1 (\mu)+1}{2}} - \mu$. This means that $\mu = q^{\frac{\nu+1}{2} - u_1(\mu)} - u_2(\mu)$.

We prove it by decreasing induction on $\mu$, for $q^{\frac{\nu-1}{2}} \geq \mu \geq 1$.
For the basis step we have $\mu = q^{\frac{\nu-1}{2}}$, thus $u_1(\mu) = 1$ and $u_2(\mu) = 0$. According to our claim, $Z(H(Q),\mu,\mu)$ is equal to $\mu + g - 1$, which has been already proven in the first part of this proposition.
For the inductive step, we now assume that our claim is true for $Z(H(Q),\mu,\mu)$ and we want to prove it for $Z(H(Q),\mu-1,\mu-1)$.
We note that:
\begin{eqnarray*}
Z(H(Q),\mu,\mu) & = & \#\{\alpha \in \cup_{v=1}^{\mu-2}(-v+H(Q))\mid \alpha \notin H(Q)\} + \\ 
& & \#\{\alpha \in -(\mu-1)+H(Q))\mid \alpha \notin \cup_{v=0}^{\mu-2}(-v+H(Q))\} \\
& = & Z(H(Q),\mu-1,\mu-1) + \#T(\mu),
\end{eqnarray*}
where $T(\mu) = \{\alpha \in -(\mu-1)+H(Q)\mid \alpha \notin \cup_{v=0}^{\mu-2}(-v+H(Q))\}$. Thus $Z(H(Q),\mu-1,\mu-1) = Z(H(Q),\mu,\mu) - \#T(\mu)$.
We consider two cases: $\mu$ such that $q^{\frac{\nu-2u_1(\mu-1)-1}{2}} < \mu-1 < q^{\frac{\nu-2u_1(\mu-1)+1}{2}}$ and  $\mu-1 = q^{\frac{\nu-2u_1(\mu-1)+1}{2}}-1$

Let us consider the first case, $\mu$ such that $q^{\frac{\nu-2u_1(\mu-1)-1}{2}} < \mu-1 < q^{\frac{\nu-2u_1(\mu-1)+1}{2}}$, then $u_1(\mu-1) = u_1(\mu)$ and $u_2(\mu-1) = u_2(\mu) + 1$.

By induction we have that $Z(H(Q), \mu, \mu) = \mu + g - 1 - (q^{\frac{\nu-1}{2}}\sum_{i=1}^{u_1(\mu)-1}(q^{1-\frac{i}{2}} - q^{-\frac{i}{2}}) + u_2(\mu)q^{\frac{u_1(\mu)}{2}})$. We claim that $\#T(\mu) = q^{\frac{u_1(\mu)}{2}}$. By (5) in Lemma \ref{rem1}, for any $x,y \in \cup_{i=0}^{u_1(\mu-1)} S^i_\nu$, $x > y$ we have that $x - y \geq q^{\frac{\nu-2u_1(\mu-1)+1}{2}}$, moreover $\mu-1 < q^{\frac{\nu-2u_1(\mu-1)+1}{2}}$. Therefore, one has that $(-(\mu-1) + \cup_{i=0}^{u_1(\mu-1)} S^i_\nu) \cap (\cup_{v=0}^{\mu-2} - v + H(Q)) = \emptyset$. Then $-(\mu-1) + \cup_{i=0}^{u_1(\mu-1)} S^i_\nu \subseteq T(\mu)$. Actually, the previous inclusion is an equality. We shall prove it by contradiction: we assume that there exists an element $x \in T(\mu)$ but not in $-(\mu-1) + \cup_{i=0}^{u_1(\mu-1)} S^i_\nu$. By definition of $T(\mu)$, we have that
$x \in -(\mu-1) + (S^{\infty}_\nu \cup (\cup_{i=u_1(\mu-1)+1}^{j/2} S^i_\nu))$ where $j=2\lfloor \nu/2 \rfloor$. Consider $y < x$ to be the previous element of $x$ in $-(\mu-1) + H(Q)$.

By (3) and (4)  in Lemma \ref{rem1}, we have that 
$$x - y \leq q^{\frac{\nu-2u_1(\mu-1)-1}{2}} < \mu - 1.$$

Thus, $-(\mu-2) \leq -(\mu-1)+(x-y) \leq 0$ and
$$x \in -(\mu-1)+(x-y)+H(Q) \subseteq \cup_{v=0}^{\mu-2}(-v+H(Q)).$$

This means $x \notin T(\mu)$, which is a contradiction.
It follows that $\# T(\mu) = \# \big(-(\mu-1) + \cup_{i=0}^{u_1(\mu-1)} S^i_\nu\big) = \# \big(\cup_{i=0}^{u_1(\mu-1)} S^i_\nu\big) = q^{\frac{u_1(\mu-1)}{2}} = q^{\frac{u_1(\mu)}{2}}$.

We consider now the second case,  $\mu-1 = q^{\frac{\nu-2u_1(\mu-1)+1}{2}}-1$, then $u_1(\mu-1) = u_1(\mu)+1$ and $u_2(\mu-1) = 0$. By using the same argument as in the first case, one may also prove that $\#T(\mu) = q^{\frac{u_1(\mu)}{2}}$.

\end{proof}

\begin{cor} \label{coromu}
By using the same notation of the previous proposition,  for $2g \leq \mu_2 < \mu_1 < n-1$ and $\ell \geq q^{\frac{\nu-1}{2}}$ we have that $M_{\ell}(C_1,C_2) = n - \dim C_2$.
\end{cor}

\begin{proof}
By \ref{propmu}, $M_{\ell}(C_1,C_2) \geq n - \dim C_2$. And $M_{\ell}(C_1,C_2) \leq n - \dim C_2$, by the Singleton bound for one-point algebraic geometric codes and the result holds.
\end{proof}

This means that for $\ell \geq q^{\frac{\nu-1}{2}}$ the Singleton bound is reached. Note that for $\ell < q^{\frac{\nu-1}{2}}$, the bound in Proposition \ref{propmu} allows us to get a refined bound since we could consider $h_{c-m}$.

As before, this result has an asymptotically implication:

\begin{cor} \label{cormu2}
Let $q$ be an even power of a prime, $\frac{2}{\sqrt{q}-1} \leq \tilde{R}_2 \leq \tilde{R}_1 < 1$, and $\tilde{R} = \tilde{R}_1 - \tilde{R}_2$. There exists a sequence of one-point algebraic geometric codes $C_{2,i} = C_{\mathcal{L}}(D_i,\mu_{2,i}Q) \subsetneq C_{1,i} = C_{\mathcal{L}}(D_i,\mu_{1,i}Q)$, $\mu_i = \mu_{1,i}-\mu_{2,i}$, such that:
$n_i = n(C_{2,i}) = n(C_{1,i}) \rightarrow \infty$, $\mu_{j,i} / n_i
\rightarrow \tilde{R}_j$ when $i \rightarrow \infty$, for $j = 1,2$.
Let $\ell_i = \dim C_{1,i} - \dim C_{2,i}$, $M = \liminf
M_{\ell_i}(C_{1,i},C_{2,i}) / n_i$,  $R_j = \lim \frac{\dim C_{i,j}}{n_i}$ for $j = 1,2$, and $R = R_1 - R_2$, we have that:
\begin{equation*}M = 1 - R_2 \quad \mbox{ if } \quad R \geq \frac{1}{q - \sqrt{q}}\end{equation*}
and
\begin{equation*}M \geq 1 - R_2  - \Bigg( \frac{1}{q - \sqrt{q}} \Bigg( \sum_{i=1}^{- \lfloor\log_q(R(1-\frac{1}{\sqrt{q}}))\rfloor - 1} (q^{1-\frac{i}{2}} - q^{-\frac{i}{2}}) +  \end{equation*}
\begin{equation*} + q^{1+\frac{1}{2}  \lfloor\log_q(R(1-\frac{1}{\sqrt{q}}))\rfloor} \Bigg) + Rq^{-\frac{1}{2} \lfloor\log_q(R(1-\frac{1}{\sqrt{q}}))\rfloor}\Bigg) \quad \mbox{ if } \quad R <\frac{1}{q - \sqrt{q}}.\end{equation*}
\end{cor}

\begin{proof}
Let $(\mathcal{F}_1,\mathcal{F}_2,\ldots)$ be the tower of function fields defined in section \ref{secintroduction}, and $0 \leq \mu_{2,i} < \mu_{1,i} \leq n_i - 1$ with  $\mu_{j,i} / n_i \rightarrow \tilde{R}_j$ for $j = 1,2$, where $n_i$ is the length of rational places of $\mathcal{F}_i$.\\
Now consider $C_{j,i} = C_{\mathcal{L}}(D_i,\mu_{j,i}Q)$ for $j = 1,2$, where $D_i$ is a divisor of degree $n_i - 1$, with $n_i - 1$ distinct places not containing $Q$, which is the unique pole of $x_1 \in \mathcal{F}_i$.
Since we assume that $\frac{2}{\sqrt{q}-1} \leq \tilde{R}_2 \leq \tilde{R}_1 < 1$ then $R_j = \tilde{R}_j - \frac{1}{\sqrt{q}-1}$, for $j=1,2$ and $R = \tilde{R}$. By taking the limit of the result obtained in Proposition \ref{propmu} and Corollary \ref{coromu} the result holds.\\
\end{proof}

Note that if $\lfloor\log_q(R(1-\frac{1}{\sqrt{q}}))\rfloor = \log_q(R(1-\frac{1}{\sqrt{q}}))$ then the formulas in Corollary \ref{cormu2} become:
$$M = 1 - R_1 \quad \mbox{if} \quad R \geq \frac{1}{q - \sqrt{q}}$$
and
$$M \geq 1 - R_1 - \frac{1}{q - \sqrt{q}}  \sum_{i=1}^{- \log_q(R(1-\frac{1}{\sqrt{q}})) - 1} \left(q^{1-\frac{i}{2}} - q^{-\frac{i}{2}} \right) \mbox{ if } R <\frac{1}{q - \sqrt{q}}.
$$

Corollary \ref{coroAG} can be used for $\rho \le \min \{ R,\frac{1}{\sqrt{q}-1}\}$. If $C_{2,i} = \{ 0\}$ for all $i$, then the value $M$ of Corollary \ref{cormu2} represents the asymptotic value of the highest GHW of $C_{i,1}$. Note that Corollary \ref{cormu2} can be used for  any value of $R$, but \ref{coroAG} cannot.
 
\bibliography{biblio}
\bibliographystyle{plain}

\end{document}